\documentclass[conference]{IEEEtran}

\usepackage{cite}
\usepackage{amsmath,amssymb,amsfonts}
\usepackage{algorithmic}
\usepackage{graphicx}
\usepackage{textcomp}
\usepackage{xcolor}
\usepackage{auxhook}
\usepackage{amsthm}
\usepackage{mathtools}
\usepackage{amssymb}

\theoremstyle{definition}
\newtheorem{definition}{Definition}

\newtheorem{theorem}{Theorem}
\newtheorem{corollary}{Corollary}

\begin{document}

\title{Incremental Model Transformations with Triple Graph Grammars for Multi-version Models
\thanks{This work was developed mainly in the course of the project modular and incremental Global Model Management (project number 336677879) funded by the DFG.}
}

\author {
	\IEEEauthorblockN{
		1\textsuperscript{st} Matthias Barkowsky
	}
	\IEEEauthorblockA{
		\textit{System Analysis and Modeling Group} \\
		\textit{Hasso-Plattner Institute at the University of Potsdam}\\
		Prof.-Dr.-Helmert-Str. 2-3, D-14482 Potsdam, Germany \\
		matthias.barkowsky@hpi.de
	}
	\and
	\IEEEauthorblockN {
		2\textsuperscript{nd} Holger Giese
	}
	\IEEEauthorblockA {
		\textit{System Analysis and Modeling Group} \\
		\textit{Hasso-Plattner Institute at the University of Potsdam}\\
		Prof.-Dr.-Helmert-Str. 2-3, D-14482 Potsdam, Germany \\
		holger.giese@hpi.de
	}
}

\maketitle

\begin{abstract}
Like conventional software projects, projects in model-driven software engineering require adequate management of multiple versions of development artifacts, importantly allowing living with temporary inconsistencies. In previous work, multi-version models for model-driven software engineering have been introduced, which allow checking well-formedness and finding merge conflicts for multiple versions of a model at once. However, also for multi-version models, situations where different artifacts, that is, different models, are linked via automatic model transformations have to be handled.
 
In this paper, we propose a technique for jointly handling the transformation of multiple versions of a source model into corresponding versions of a target model, which enables the use of a more compact representation that may afford improved execution time of both the transformation and further analysis operations. Our approach is based on the well-known formalism of triple graph grammars and the aforementioned encoding of model version histories called multi-version models. In addition to batch transformation of an entire model version history, the technique also covers incremental synchronization of changes in the framework of multi-version models.
 
We show the correctness of our approach with respect to the standard semantics of triple graph grammars and conduct an empirical evaluation to investigate the performance of our technique regarding execution time and memory consumption. Our results indicate that the proposed technique affords lower memory consumption and may improve execution time for batch transformation of large version histories, but can also come with computational overhead in unfavorable cases.
\end{abstract}


\section{Introduction} \label{sec:introduction}

In model-driven software development, models are treated as primary development artifacts. Complex projects can involve multiple models, which describe the system under development at different levels of abstraction or with respect to different system aspects and can be edited independently by a team of developers. In this case, consistency of the holistic system description is ensured by batch model transformations, which automatically derive new models from existing ones, and incremental model transformations, that is, model synchronizations, which propagate changes to a transformation's source model to the transformation's target model\cite{seibel2010dynamic}.

Similarly to program code, the evolution of models via changes by different developers requires management of the resulting versions of the software description. In particular, version management has to support parallel development activities of multiple developers working on the same artifact, where living with inconsistencies may temporarily be necessary to avoid loss of information \cite{Finkelstein+1994}. In \cite{Barkowsky2022}, we have introduced multi-version models as a means of managing multiple versions of the same model that also enables monitoring the consistency of the individual model versions and potential merge results of versions developed in parallel.

However, with model transformations effectively linking multiple models via consistency relationships, considering only the evolution of a single model without its context is insufficient for larger model-driven software development projects. Thus, a mechanism for establishing consistency of different versions of such linked models that allows parallel development of multiple versions is required. On the one hand, this requires support for transforming multiple versions of a source model into the corresponding target model versions, for instance when a new kind of model is introduced to the development process. On the other hand, the inherently incremental scenario of development with multiple versions calls for efficient synchronization of changes between versions of related models. In order to achieve efficient transformation and synchronization of models with multiple versions and enable further analysis operations as described in \cite{Barkowsky2022}, a close integration of this transformation and synchronization mechanism and version handling seems desirable.

Therefore, in this paper\footnote{See \cite{report} for the technical report version of the paper.}, we explore a step in the direction of model transformations for multi-version models by adapting the well-known formalism of triple graph grammars, which enables the implementation of single-version model transformations and synchronizations, to the multi-version case.

The remainder of the paper is structured as follows: In Section \ref{sec:preliminaries}, we reiterate the concepts of graphs, graph transformations, triple graph grammars, and multi-version models. We subsequently present our approach for deriving transformation rules that work on multi-version models from single-version transformation specifications in the form of triple graph grammars in Section \ref{sec:mv_rules_derivation}. In Section \ref{sec:mv_transformation_execution}, we describe how the derived rules can be used to realize joint transformation of all model versions encoded in a multi-version model and prove the technique's correctness with respect to the semantics of triple graph grammars. Section \ref{sec:mv_incremental_execution} discusses the extension of the approach to incremental model synchronization. Section \ref{sec:evaluation} reports on results of an initial evaluation of the solution's performance based on an application scenario in the software development domain. Related work is discussed in Section \ref{sec:related_work}, before Section \ref{sec:conclusion} concludes the paper.

\section{Preliminaries} \label{sec:preliminaries}

In this section, we give an overview of required preliminaries regarding graphs and graph transformations, triple graph grammars, and multi-version models.

\subsection{Graphs and Graph Transformations}

We briefly reiterate the concepts of graphs, graph morphisms and graph transformations and their typed analogs as defined in \cite{Ehrig+2006} and required in the remainder of the paper.

A graph $G = (V^G, E^G, s^G, t^G)$ is given by a set of nodes $V^G$, a set of edges $E^G$ and two functions $s^G: E^G \rightarrow V^G$ and $t^G: E^G \rightarrow V^G$ assigning each edge a source and target node. A graph morphism $m: G \rightarrow H$ consists of two functions $m^V: V^G \rightarrow V^H$ and $m^E: E^G \rightarrow E^H$ such that $s^H \circ m^E = m^V \circ s^G$ and $t^H \circ m^E = m^V \circ t^G$. We call $m^V$ the \emph{vertex morphism} and $m^E$ the \emph{edge morphism}.

A typed graph $G^T = (G, \mathit{type}^G)$ comprises a graph $G$ along with a typing morphism $\mathit{type}: G \rightarrow TG$ into a type graph $TG$. In this paper, we consider a model to be a typed graph, with the type graph defining a modeling language by acting as a metamodel. A typed graph morphism from a typed graph $G^T = (G, \mathit{type}^G)$ into a typed graph $H^T = (H, \mathit{type}^H)$ with the same type graph is a graph morphism $m^T: G \rightarrow H$ such that $\mathit{type}^G = \mathit{type}^H \circ m^T$. A (typed) graph morphism $m$ with injective functions $m^V$ and $m^E$ is called a monomorphism. If $m^V$ and $m^E$ are also surjective, $m$ is called an isomorphism.

Figure \ref{fig:example_graph_type_graph} shows an example typed graph on the left along with the corresponding type graph on the right. The typing morphism is encoded by the node's labels. The graph represents an abstract syntax graph of a program written in an object-oriented programming language. Nodes may represent class declarations ($ClassDecl$), field declarations ($FieldDecl$) or type accesses ($TypeAccess$). Class declarations can contain field declarations via edges of type $declaration$, whereas field declarations can reference a class declaration as the field type via a $TypeAccess$ node and edges of type $access$ and $type$. The graph contains two class declarations, one of which contains a field declaration, the field type of which is given by the other class declaration.

\begin{figure}
\centering
\includegraphics[height=0.2\linewidth]{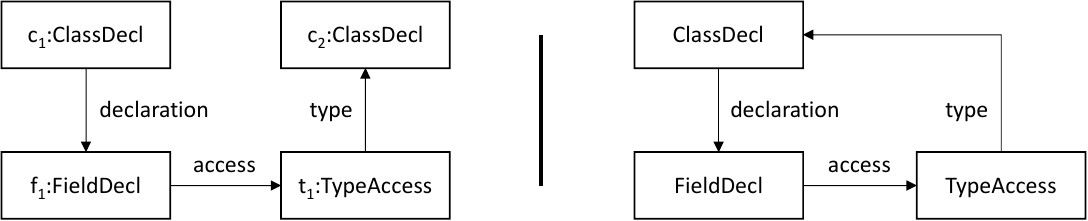}
\caption{Example graph (left) and type graph (right)} \label{fig:example_graph_type_graph}
\end{figure}

A (typed) graph transformation rule $\gamma$ is characterized by a span of (typed) graph monomorphisms $L \xleftarrow{l} K \xrightarrow{r} R$ and can be applied to a graph $G$ via a monomorphism $m : L \rightarrow G$ called match that satisfies the so-called dangling condition \cite{Ehrig+2006}. The result graph $H$ of the rule application is then formally defined by a double pushout over an intermediate graph \cite{Ehrig+2006}. We denote the application of $\gamma$ via $m$ by $G \rightarrow^\gamma_m H$. Intuitively, the application of $\gamma$ deletes the elements in $m(L)$ that do not have a corresponding element in $R$ and creates new elements for elements in $R$ that do not have a corresponding element in $L$. The graph $L$ is called the rule's \emph{left-hand side}, $K$ is called the rule's \emph{glueing graph}, and $R$ is called the \emph{right-hand side}.

$\gamma$ is a graph production if it does not delete any elements, that is, $l$ is surjective. In this case, since $L$ and $K$ are isomorphic, we also use the simplified representation $L \xrightarrow{r} R$.

Figure \ref{fig:example_gt_rule} shows an example graph production in shorthand notation, where preserved elements are colored black, whereas created elements are colored green and marked by an additional ``++'' label. For two existing classes, the production creates a field declaration in one of them that references the other class declaration as the field type.

\begin{figure}
\centering
\includegraphics[height=0.2\linewidth]{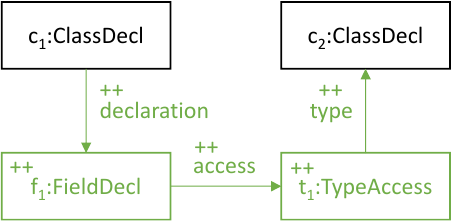}
\caption{Example graph transformation rule in shorthand notation} \label{fig:example_gt_rule}
\end{figure}

We denote a sequence of applications of rules from a set of rules $\Gamma$ to a graph $G$ with resulting graph $G'$ by $G \rightarrow^{\Gamma} G'$. We say that such a rule application sequence is maximal if it cannot be extended by any application of a rule from $\Gamma$.

\begin{definition} \textit{Maximal Rule Application Sequence}
A sequence of rule applications $G \rightarrow^{\Gamma} G'$ with a set of graph transformation rules $\Gamma$ is maximal if no rule in $\Gamma$ is applicable to $G'$.
\end{definition}

\subsection{Triple Graph Grammars}

Triple graph grammars (TGGs) were initially presented by Schuerr \cite{Sch94_2_ref}. This paper is based on the slightly adapted version introduced in \cite{Giese+2014}. A TGG relates a source and a target modeling language via a correspondence modeling language and is characterized by a set of TGG rules. In \cite{Giese+2014}, a TGG rule is defined by a graph production that simultaneously transforms connected graphs from the source, correspondence and target modeling language into a consistently modified graph triplet. The set of TGG rules has to include an axiom rule, which has a triplet of empty graphs as its left-hand side and defines a triplet of starting graphs via its right-hand side.

The left-hand side of a TGG rule $\gamma = L \xrightarrow{r} R$ can be divided into the source, correspondence, and target domains $L_S$, $L_C$, and $L_T$ respectively, with $L_S \subseteq L$, $L_C \subseteq L$, and $L_R \subseteq L$ and $L_S \uplus L_C \uplus L_R = L$. The right-hand side can similarly be divided into three domains $R_S$, $R_C$, and $R_T$. The type graph for graph triplets and TGG rules is hence given by the union of the type graphs defining the source, correspondence, and target language along with additional edges connecting nodes in the correspondence language to elements in the source and target language. It is also assumed that each element in $L_S$, $L_T$, $R_S$, and $R_T$ is connected to exactly one node in $L_C$ or $R_C$ and that each rule creates exactly one node in the correspondence domain.

TGGs can be employed to transform a model of the source language into a model of the target language. This requires the derivation of so-called forward rules from the set of TGG rules. A forward rule for a TGG rule $\gamma = L \xrightarrow{r} R$ can be constructed as $\gamma^F = L^F \xleftarrow{id} L^F \xrightarrow{r^F} R$, where $L^F = L \cup (R_S \setminus r(L))$ and $r^F = r \cup id$, with $id$ the identity morphism. Intuitively, $\gamma^F$ already requires the existence of the elements in the source domain that would be created by an application of $\gamma$ and only creates elements in the correspondence and target domain. In the following, we also denote the subgraph of a forward rule that corresponds to the subgraph that is newly transformed by the rule by $L^T = L^F \setminus L$.

Additionally, the derivation of a forward rule requires a technical extension to avoid redundant translation of the same element. Therefore, a dedicated \emph{bookkeeping node}, which is connected to every currently untranslated source element via a \emph{bookkeeping edge}, is introduced. Then, a bookkeeping node and bookkeeping edges to all elements in $L^T$ are added to the forward rule's left-hand side. The bookkeeping node is also added to the rule's glueing graph and right-hand side. The application of the forward rule via $m$ thus requires that elements in $m(L^T)$ are untranslated, as indicated by the presence of bookkeeping edges, and marks them as translated by deleting the adjacent bookkeeping edges.

Note that, in order to allow bookkeeping edges and outgoing edges of correspondence nodes to target regular edges, a slightly extended graph model is used, which is detailed in \cite{giese2010toward}. In this paper, we will call such graphs graphs with bookkeeping. We say that two graphs with bookkeeping $G$ and $H$ are equal up to isomorphism including bookkeeping if and only if there exists an isomorphism $iso : G \rightarrow H$, which entails that for all nodes and edges $x \in G$, it holds that $\exists b \in E'^{G}: t'^{G}(b) = x \leftrightarrow \exists b' \in E'^{H}: t'^{H}(b') = iso(x)$, where $E'^{G}$ denotes the set of bookkeeping edges and $t'^{G}$ the related target function.

Figure \ref{fig:example_tgg_rule} shows a TGG rule for linking the language for abstract syntax graphs given by the type graph in Figure \ref{fig:example_graph_type_graph} to a modeling language for class diagrams given by the type graph on the right in Figure \ref{fig:example_type_graphs_tgg}, using the correspondence language on the left. The rule simultaneously creates a $FieldDecl$ and $TypeAccess$ along with associated edges in the source domain (labeled S) and a corresponding $Association$ with associated edges in the target domain (labeled T), which are linked via a newly created correspondence node of type $CorrField$ in the correspondence domain (labeled C). Edges from the correspondence node to other edges are omitted for readability.

\begin{figure}
\centering
\includegraphics[width=0.85\linewidth]{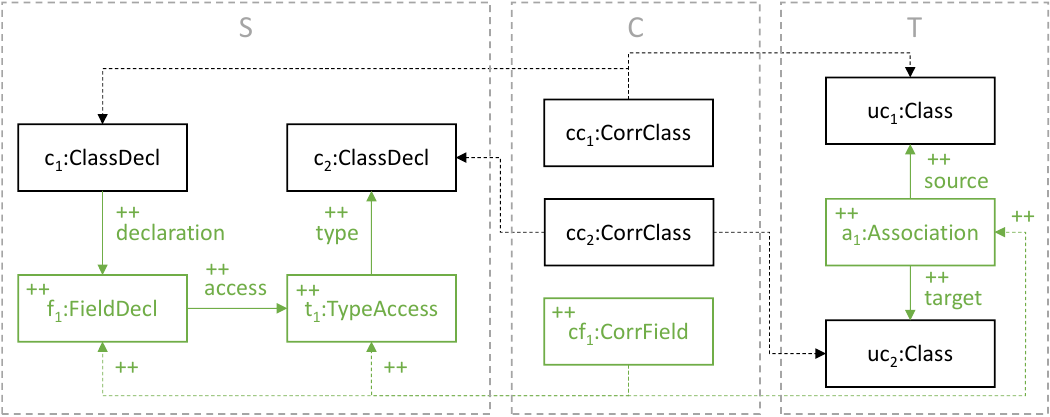}
\caption{Example TGG rule in shorthand notation} \label{fig:example_tgg_rule}
\end{figure}

\begin{figure}
\centering
\includegraphics[width=0.4\linewidth]{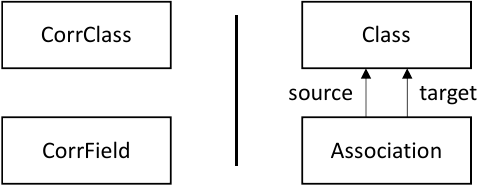}
\caption{Example type graphs for the TGG rule in Figure \ref{fig:example_tgg_rule}} \label{fig:example_type_graphs_tgg}
\end{figure}

Figure \ref{fig:example_forward_rule} shows the forward rule derived from the TGG rule in Figure \ref{fig:example_tgg_rule}. The elements $f_1$ and $t_1$ and adjacent edges are no longer created but preserved instead. Also, the rule contains a bookkeeping node and adjacent bookkeeping edges to these elements. The rule's application then deletes these bookkeeping edges and creates the corresponding elements in the target domain along with the linking node $cf_1$ and edges in the correspondence domain. The bookkeeping mechanism is however not visualized for readability reasons. 

\begin{figure}
\centering
\includegraphics[width=0.85\linewidth]{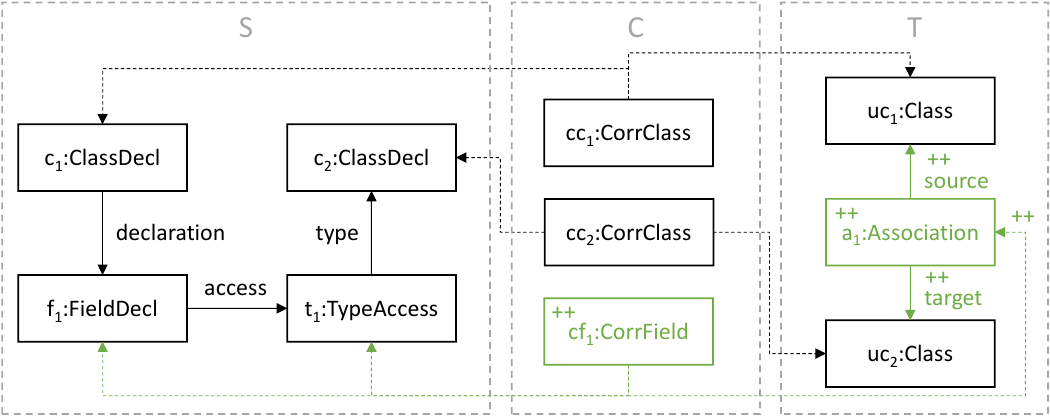}
\caption{Example forward rule derived from the TGG rule in Figure \ref{fig:example_tgg_rule}, with the bookkeeping mechanism omitted for readability reasons} \label{fig:example_forward_rule}
\end{figure}

TGGs can also be used to perform a transformation from the target to the source language by means of similarly derived backward rules. In the following, we will focus on the forward case. However, the backward case simply works analogously.

A TGG without any critical pairs \cite{Ehrig+2006} among its rules is called \emph{deterministic} \cite{Giese+2014}. A forward transformation with a deterministic TGG can be executed via an operation $trans^{F}$, which adds a bookkeeping node and bookkeeping edges to all elements in the soruce mode and then applies the TGG's forward rules for as long as there is a match for any of them. For a deterministic TGG with a set of forward rules $\Gamma$ and a starting model triplet $SCT$, any produced maximal rule application sequence $SCT \rightarrow^{\Gamma} SCT'$ then constitutes a correct model transformation and yields the same result if it deletes all bookkeeping edges in $SCT$. In this paper, we will focus on such deterministic TGGs, which allow for efficient practical implementations that avoid potentially expensive undoing of forward rule applications and backtracking \cite{Giese+2014}.

In addition to a full batch transformation of a previously untransformed model, TGGs also enable incremental synchronization of changes to an already transformed model to the transformation result. In the most basic case, this involves undoing all forward rule applications that are invalidated by the deletion of elements in a first step. In a second step, elements that are no longer covered as well as newly created elements are then transformed via the TGG's forward rules.

\subsection{Multi-version Models}

In this paper, we consider models in the form of typed graphs. A model modification can in this context be represented by a span of morphisms $M \leftarrow K \rightarrow M'$, where $M$ is the original model, which is modified into a changed model $M'$ via an intermediate model $K$ \cite{taentzer2014fundamental}. A \emph{version history} of a model is then given by a set of model modifications $\Delta^{M_{\{1,...,n\}}}$ between models $M_1, M_2, ..., M_n$ with type graph $TM$. We call a version history with a unique initial version and acyclic model modification relationships between the individual versions a \emph{correct} version history.

In \cite{Barkowsky2022}, we have introduced \emph{multi-version models} as a means of encoding such a version history in a single consolidated graph. Therefore, an adapted version of $TM$, $TM_{mv}$, is created. To represent model structure, $TM_{mv}$ contains a node for each node and each edge in $TM$. Source and target relationships of edges in $TM$ are represented by edges in $TM_{mv}$. In addition, a $version$ node with a reflexive $suc$ edge is added to $TM_{mv}$, which allows the materialization of the version history's version graph. The version graph and the model structure are linked via $cv_v$ and $dv_v$ edges from each node $v$ in $TM_{mv}$ to the $version$ node.

Figure \ref{fig:example_type_graph_adapted} displays the adaptation of the type graph from Figure \ref{fig:example_graph_type_graph}. $cv$ and $dv$ edges are omitted for readability reasons.

\begin{figure}
\centering
\includegraphics[width=0.9\linewidth]{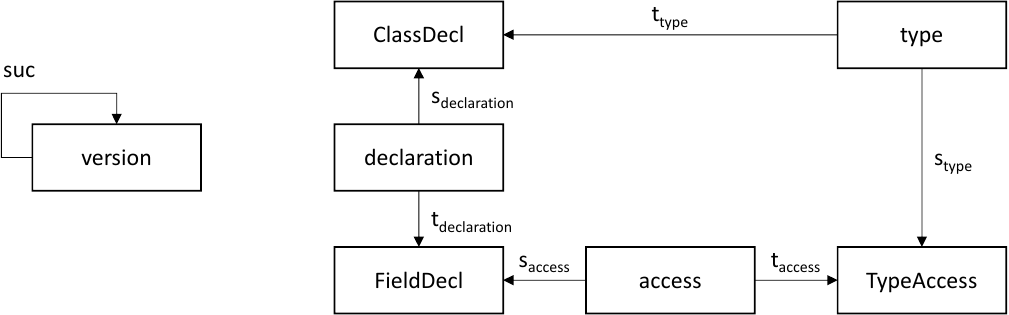}
\caption{example adapted type graph derived from the type graph in Figure \ref{fig:example_graph_type_graph}, with $cv$ and $dv$ edges omitted for readability reasons} \label{fig:example_type_graph_adapted}
\end{figure}

$TM_{mv}$ allows the translation of $\Delta^{M_{\{1,...,n\}}}$ into a single typed graph $MVM$ conforming to $TM_{mv}$, which is called a \emph{multi-version model}, via a procedure $comb$. This translation yields a bijective function $origin: V^{MVM} \rightarrow \bigcup_{i \in \{1, 2, ..., n\}} V^{M_i} \cup E^{M_i}$ mapping the nodes in $MVM$ to their respective original element. An individual model version can be extracted from $MVM$ via the projection operation $proj(MVM, i) = M_i$. Finally for a node $v_{mv} \in V^{MVM}$, the set of model versions that include the element $origin(v_{mv})$ can be computed via the function $p$, with $p(v_{mv}) = \{M_i \in \{M_1, M_2, ..., M_n\} |  origin(v_{mv}) \in M_i\}$.

\section{Derivation of Multi-version Transformation Rules from Triple Graph Grammars} \label{sec:mv_rules_derivation}

The transformation of the individual model versions encoded in a multi-version model with a TGG can trivially be realized via the projection operation $proj$. However, the multi-version model may in practice afford a more compact representation compared to an explicit enumeration of all model versions, as derived via $proj$.

In such practical application scenarios, operations concerning all model versions that directly work on the multi-version model may therefore also perform better regarding execution time than the corresponding operations on individual model versions, as has already been demonstrated for the case of pattern matching in \cite{Barkowsky2022}. Since pattern matching also constitutes an important task in model transformation via TGGs, a direct, joint translation of all model versions based on the multi-version model representation seems desirable.

Given a TGG, graph transformation rules for the joint translation of all source or target model versions encoded in a multi-version model can be derived from the regular translation rules in a straightforward manner. In the following, we will discuss the deriviation for forward translation. Rules for the backward case can be derived analogously.

First, the adapted multi-version type graph for the TGG's merged source, correspondence and target type graph is created via the translation procedure described in \cite{Barkowsky2022}. However, edges between the correspondence and source or target type graph are simply translated as edges rather than nodes. The resulting adapted type graph $TG_{mv}$ for multi-version models is also extended by two additional edges, $ucv_v$ and $udv_v$, for each node $v$ in the source domain of the merged type graph. Source and target of these edges are given by $s^{TG_{mv}}(ucv_v) = s^{TG_{mv}}(udv_v) = v$ and $t^{TG_{mv}}(ucv_v) = t^{TG_{mv}}(udv_v) = version$, where $version$ is the dedicated version node.

Analogously to the bookkeeping edges in the original type graph, these edges will be used to encode in which versions an element represented by a node $v_{mv}$ with type $v$ has already been translated. We therefore define the set of versions where $v_{mv}$ has not been translated yet $u(v_{mv})$ analogously to the set of versions $p(v_{mv})$ where $v_{mv}$ is present  \cite{Barkowsky2022}, except that $ucv_v$ and $udv_v$ replace $cv_v$ and $dv_v$ in the definition.

Then, for each forward rule $\gamma = L \xleftarrow{l} K \xrightarrow{r} R$ a corresponding multi-version forward rule is created via a procedure $adapt$, with $adapt(\gamma) = trans'(L) \xleftarrow{l_{mv}} trans'(K) \xrightarrow{r_{mv}} trans'(R)$. The vertex morphism of $l_{mv}$ is given by $l_{mv}^V = origin^{-1}\,\circ\,l\,\circ\,origin$ and the edge morphisms by $l_{mv}^E = s \circ origin^{-1} \circ l^E \circ origin \circ s^{-1}$ and $l_{mv}^E = t \circ origin^{-1} \circ l^E \circ origin \circ t^{-1}$ for all edges representing source respectively target relationships. $r_{mv}$ is constructed analogously.

The $trans'$ procedure is a minor adaptation of the $trans$ procedure in \cite{Barkowsky2022}, which ignores the bookkeeping node and bookkeeping edges, and translates correspondence edges to edges rather than nodes, but otherwise works analogously. The bookkeeping mechanism is translated into the additional constraint $P \neq \emptyset$ over $trans'(L)$, where $P = (\bigcap_{v_{mv} \in V^{trans'(L)}} p(v_{mv}) \cap \bigcap_{v_{mv} \in origin^{-1}(L^T)} u(v_{mv})) \setminus \bigcup_{v_{mv} \in V^{trans'(L)} \setminus origin^{-1}(L^T)} u(v_{mv})$.

The application of the adapted rule additionally creates outgoing $cv$ and $dv$ edges for all nodes $v^C_{mv} \in V^{trans(R)} \setminus (origin^{-1}\,\circ\,r\,\circ\,origin)(trans(K))$ to realize the assignment $p(v^C_{mv}) \coloneqq P$. Furthermore, for each $v_{mv} \in origin^{-1}(r(l^{-1}(L^T)))$, the application also adds and deletes outgoing $ucv$ and $udv$ edges to realize the modification $u(v_{mv}) \coloneqq u(v_{mv}) \setminus P$. Note that, since the computation of the $p$ and $u$ sets requires considedring paths of arbitrary length, these computations cannot technically be defined as part of the graph transformation but have to be realized externally.

For a set of forward rules $\Gamma$, $adapt(\Gamma) = \{adapt(\gamma) | \gamma \in \Gamma\}$ denotes the corresponding set of multi-version forward rules.

\section{Execution of Multi-version Transformations (Detailed)} \label{sec:mv_transformation_execution}

The forward transformation of all model versions in a multi-version model $MVM$ according to a TGG can jointly be performed via the TGG's set of multi-version forward rules.

In a first step, all $ucv$ and $udv$ edges in $MVM$ are removed. Then, for each edge $e_{cv} \in E^{MVM}$ with $type(e_{cv}) = cv_{x}$ and $s^{MVM}(e_{cv})$, an edge $e_{ucv}$ with $type(e_{ucv}) = ucv_{x}$ and $s^{MVM}(e_{cv}) = s^{MVM}(e_{ucv})$ and $t^{MVM}(e_{cv}) = t^{MVM}(e_{ucv})$ is created. For all $dv$ edges, corresponding $udv$ edges are created analogously. Thus, after the creation of the $ucv$ and $udv$ edges, it holds that $\forall v_{vm} \in V^{MVM}: u(v_{vm}) = p(v_{vm})$.

Subsequently, the simultaneous transformation of all model versions encoded in $MVM$ is performed similarly to the regular transformation of a single model version via the TGG. More specifically, the adapted forward rules of the TGG are applied to $MVM$ until no such rule is applicable anymore.

In the following, we will argue that this transformation approach is correct in the sense that it yields the same result as the transformation of the individual model versions via regular forward rules. Therefore, we first extend the projection operation $proj$ from \cite{Barkowsky2022} to a bookkeeping-sensitive variant.

\begin{definition} \textit{(Bookkeeping-sensitive Projection)}
For a multi-version model $MVM$ with version graph $V$ and model version $M_t$ with corresponding $m_t \in V^V$, the bookkeeping-sensitive projection operation works similarly to the regular projection operation $proj$, except that it also adds a bookkeeping node and bookkeeping edges to an element $origin(v)$ iff $M_t \in u(v)$ for all $v \in V^{MVM}$. We also denote the result of the bookkeeping-sensitive projection operation by $MVM[t] = proj^{M}(MVM, t)$.
\end{definition}

We also define two sets that represent the bookkeeping during the transformation process.

\begin{definition} \textit{(Bookkeeping Set)}
For a model $M$, we denote the set of translated elements (vertices and edges) by $B(M) = \{x \in M | \nexists b \in E'^{M}: t'^{M} = x\}$, with $E'^{M}$ the set of bookkeeping edges in $M$ and $t'^{M}$ the target function for bookkeeping edges. We also call $B(M)$ the \emph{bookkeeping set} of $M$.
\end{definition}

\begin{definition} \textit{(Projection Bookkeeping Set)}
For a multi-version model $MVM$ and version $t \in V^V$, with $V$ the version graph, we denote the set of already handled elements (vertices and edges) in $MVM[t]$ by $B_{mv}(MVM[t]) = \{x \in MVM[t] | t \notin u(proj^{-1}(x))\}$. We also call $B_{mv}(MVM[t])$ the \emph{projection bookkeeping set} of $MVM[t]$.
\end{definition}

The following theorem then states that, at the start of the transformation process via adapted forward rules, the prepared multi-version model correctly encodes the starting situation for the translation of the individual model versions.

\begin{theorem} \label{the:correctness_bookkeeping_projection}
Given a multi-version model $MVM$ encoding a version history with model versions $M_1, M_2, ..., M_n$ such that $\forall v_{vm} \in V^{MVM}: u(v_{vm}) = p(v_{vm})$, it holds that
$\forall t \in \{1, 2, ..., n\}: MVM[t] = init_F(M_t)$
up to isomorphism including bookkeeping, where $init_F(M_t)$ denotes the graph with bookkeeping resulting from the preparation of $M_t$ for the regular forward transformation, that is, the graph $M_t$ with an added bookkeeping node and bookkeeping edges to all elements in $M_t$.
\end{theorem}

\begin{proof}
Follows directly from the fact that $\forall t \in \{1, 2, ..., n\}: proj(MVM, t) = M_t$, which has been shown in \cite{Barkowsky2022}, and the definition of the bookkeeping-sensitive projection operation.
\end{proof}

By Theorem \ref{the:correctness_bookkeeping_projection}, we also get the following corollary:

\begin{corollary} \label{cor:correctness_bookkeeping_projection_sets}
Given a multi-version model $MVM$ encoding a version history with model versions $M_1, M_2, ..., M_n$ such that $\forall v_{vm} \in V^{MVM}: u(v_{vm}) = p(v_{vm})$, it holds that
$\forall t \in \{1, 2, ..., n\}: B_{mv}(MVM[t]) = B(init_F(M_t))$
up to isomorphism, where $init_F(SCT_t)$ denotes the graph with bookkeeping resulting from the preparation of $M_t$ for the regular forward transformation process, that is, the graph $M_t$ with an added bookkeeping node and bookkeeping edges to all elements in $M_t$.
\end{corollary}

\begin{proof}
Follows directly from Theorem \ref{the:correctness_bookkeeping_projection} and the definition of bookkeeping set and projection bookkeeping set.
\end{proof}

We now show that a multi-version rule is applicable to a multi-version model iff the corresponding regular rule is applicable to all model versions affected by the rule application.

\begin{theorem} \label{the:mv_rule_applicability}
A multi-version forward rule $\gamma_{mv} = L_{mv} \leftarrow K_{mv} \rightarrow R_{mv}$ is applicable to a multi-version model triplet $SCT_{mv}$ with bookkeeping via match $m$, if and only if for all $t \in P$, the associated original forward rule $\gamma = L \leftarrow K \rightarrow R$ is applicable to $SCT_{mv}[t]$ via match $origin(m)$, with $P = \bigcap_{v \in V^{L_{mv}}} p(m(v)) \cap \bigcap_{v \in V^{L^{T}_{mv}}} u(m(v))$.
\end{theorem}

\begin{proof}
For a version $t$, as has already been shown in \cite{Barkowsky2022}, the match $m : L_{mv} \rightarrow SCT_{mv}$ has a corresponding match $origin(m) : L \rightarrow SCT_{mv}[t]$ if and only if $t \in \bigcap_{v \in V^{L_{mv}}} p(m(v))$. Furthermore, due to the definition of $P$ and the construction of $\gamma_{mv}$, all elements in $m(origin(m)(L^{T}))$ have an adjacent bookkeeping edge in $SCT_{mv}[t]$ iff $t \in \bigcap_{v \in V^{L^{T}_{mv}}} u(m(v))$. Similarly, all elements in $m(origin(m)(L \setminus L^{T}))$ have no adjacent bookkeeping edge in $SCT_{mv}[t]$ iff $t \notin \bigcup_{v \in V^{L_{mv} \setminus L^{T}_{mv}}} u(m(v))$. Since $\gamma$ and $\gamma_{mv}$ delete no vertices, the dangling condition is trivially satisfied for $r$ and the match $trans(m)$. $\gamma_{mv}$ is hence applicable to $SCT_{mv}$ via $m$, with $t \in P$, iff $r$ is applicable to $SCT_{mv}[t]$ via $origin(m)$.
\end{proof}

We can now show the equivalence of a single multi-version rule application to a multi-version model to the application of the corresponding regular rule to all affected model versions.

\begin{theorem} \label{the:mv_rule_application}
For an application $SCT_{mv} \rightarrow^{\gamma_{mv}}_{m} SCT_{mv}'$ of a multi-version forward rule $\gamma_{mv} = L_{mv} \leftarrow K_{mv} \rightarrow R_{mv}$ with original forward rule $\gamma = L \leftarrow K \rightarrow R$ to a multi-version model triplet $SCT_{mv}$ with bookkeeping and version graph $V$ via match $m$, it holds up to isomorphism including bookkeeping that $\forall t \in P: SCT_{mv}'[t] = SCT'$, with the corresponding application $SCT_{mv}[t] \rightarrow^{r}_{origin(m)} SCT'$, and $\forall t \in V^V \setminus P: SCT_{mv}'[t] = SCT_{mv}[t]$, where $P = \bigcap_{v \in V^{L_{mv}}} p(m(v)) \cap \bigcap_{v \in V^{L^{T}_{mv}}} u(m(v))$.
\end{theorem}

\begin{proof}
Disregarding bookkeeping edges, all forward rules and thus also the adapted forward rules are productions. Due to the construction of the adapted forward rules, all elements created by the rule's application are only mv-present in $SCT_{mv}'$ for the versions in $P$. Therefore, for all remaining versions, $SCT_{mv}[t]$ contains exactly the same elements as $SCT_{mv}'[t]$. An isomorphism $iso: SCT_{mv}[t] \rightarrow SCT_{mv}'[t]$ is hence trivially given by the identity in this case. Since the application of $\gamma_{mv}$ only changes the projection bookkeeping sets for versions in $P$, $B_{mv}(SCT_{mv}'[t]) = B(SCT_{mv}[t])$ with isomorphism $iso$.

It thus holds up to isomorphism that $\forall t \in V^V \setminus P: SCT_{mv}'[t] = SCT_{mv}[t] \wedge B_{mv}(SCT_{mv}'[t]) = B(SCT_{mv}[t])$.

The application of $\gamma_{mv}$ to $SCT_{mv}$ yields a comatch $n : R_{mv} \rightarrow SCT_{mv}'$ and the associated application of $\gamma$ to $SCT_{mv}[t]$ similarly yields a comatch $n' : R \rightarrow SCT'$ for any $t \in P$.

An isomorphism $iso: SCT_{mv}'[t] \rightarrow SCT'$ can then be constructed as follows: Since $\gamma_{mv}$ is a production, $SCT_{mv}$ is a subgraph of $SCT_{mv}'$ and hence $SCT_{mv}[t]$ is also a subgraph of $SCT_{mv}'[t]$ except for bookkeeping. Since $\gamma$ is a production, $SCT_{mv}[t]$ is also a subgraph of $SCT'$. Isomorphic mappings for $SCT_{mv}[t]$ between $SCT_{mv}'[t]$ and $SCT'$ are thus simply given by the identity. This leaves only the elements in $n(R_{mv} \setminus L_{mv})$ and the elements in $n'(R \setminus L)$ unmapped. Due to the construction of $\gamma_{mv}$ being unique up to isomorphism, $n$ and $n'$ being monomorphisms, and $trans$ and $origin$ being bijections, the remaining isomorphic mappings are given by $n' \circ trans \circ n^{-1} \circ origin$. Note that for elements in $n(L_{mv})$, the definition of $iso$ via identity and $n' \circ trans \circ n^{-1} \circ origin$ is redundant but compatible.

Due to the definition of bookkeeping-sensitive projection, bookkeeping set, and projection bookkeeping set, it holds that $B(SCT_{mv}[t]) = B_{mv}(SCT_{mv}[t])$ and thus $B_{mv}(SCT_{mv}[t]) = B(SCT_{mv}[t]))$. Compared to $B_{mv}(SCT_{mv}[t])$, the application of $\gamma_{mv}$ only changes the projection bookkeeping set $B_{mv}(SCT_{mv}'[t])$ by adding the elements in $trans(m(L^{T}_{mv}))$. The modification to $B_{mv}(SCT_{mv}'[t])$ hence corresponds to the modification of the bookkeeping set $B(SCT')$ by the application of $\gamma$ via $trans(m)$ for the isomorphism $iso$ due to the construction of $\gamma_{mv}$.

It thus holds that $\forall t \in P: SCT_{mv}'[t] = SCT' \wedge B_{mv}(SCT_{mv}'[t]) = B(SCT')$.
\end{proof}

Based on Theorem \ref{the:mv_rule_application} for individual rule applications, we get the following corollary for sequences of rule applications:

\begin{corollary} \label{cor:mv_rule_application_sequence}
For a TGG with associated set of forward rules $\Gamma$ and multi-version forward rules $\Gamma_{mv}$ and a multi-version model triplet $SCT_{mv}$ with bookkeeping and version graph $V$, there is a sequence of rule applications $SCT_{mv} \rightarrow^{\Gamma_{mv}} SCT_{mv}'$ if and only if for all $t \in V^V$, there is a sequence of rule applications $SCT_{mv}[t] \rightarrow^{\Gamma} SCT'$ with $SCT_{mv}'[t] = SCT'$ up to isomorphism including bookkeeping.
\end{corollary}

\begin{proof}
We prove the corollary by induction over the length of the multi-version rule application sequence.

For the base case of application sequences of length 0, the identity morphism and empty application sequences trivially satisfy the corollary.

If there is a sequence of rule applications $SCT_{mv} \rightarrow^{\Gamma_{mv}} SCT_{mv}'$ if and only if for all $t \in V^V$, there is a sequence of rule applications $SCT_{mv}[t] \rightarrow^{\Gamma} SCT'$ with $SCT_{mv}'[t] = SCT' \wedge B_{mv}(SCT_{mv}'[t]) = B(SCT')$, by Theorem \ref{the:mv_rule_application} we have an extended multi-version sequence $SCT_{mv} \rightarrow^{\Gamma_{mv}} SCT_{mv}' \rightarrow^{\gamma_{mv}}_{m} SCT_{mv}''$ and all $t \in V^V$ if and only if for all $t \in V^V$, there is a sequence of regular rule applications $SCT_{mv}[t] \rightarrow^{\Gamma} SCT''$ with $SCT_{mv}''[t] = SCT'' \wedge B_{mv}(SCT_{mv}''[t]) = B(SCT'')$.

For all $t \in V^V \setminus P$, where $P = \bigcap_{v \in V^{L_{mv}}} p(m(v)) \cap \bigcap_{v \in V^{L^{T}_{mv}}} u(m(v))$, the corresponding regular rule application sequence $SCT_{mv}[t] \rightarrow^{\Gamma} SCT'$ and isomorphism $iso : SCT_{mv}'[t] \rightarrow SCT'$ are also valid for $SCT_{mv}''[t]$ and satisfy the condition on bookkeeping sets, since $SCT' = SCT_{mv}'[t] = SCT_{mv}''[t]$ (up to isomorphism).

In accordance with Theorem \ref{the:mv_rule_application}, there is an extended sequence $SCT_{mv} \rightarrow^{\Gamma_{mv}} SCT_{mv}' \rightarrow^{\gamma_{mv}}_{m} SCT_{mv}''$ if and only if for all $t \in P$, the regular rule application sequence $SCT_{mv}[t] \rightarrow^{\Gamma} SCT_{mv}'[t]$ can be extended by a rule application $SCT_{mv}'[t] \rightarrow^{\gamma}_{trans(m)} SCT_{mv}''[t]$ that satisfies the condition on bookkeeping sets.

Thus, there is a sequence of rule applications $SCT_{mv} \rightarrow^{\Gamma_{mv}} SCT_{mv}' \rightarrow^{\gamma_{mv}}_{m} SCT_{mv}''$ if and only if for all $t \in V^V$, there is a sequence of rule applications $SCT_{mv}[t] \rightarrow^{\Gamma} SCT''$ with $SCT_{mv}''[t] = SCT'' \wedge B_{mv}(SCT_{mv}''[t]) = B(SCT'')$.

With the proof for the base case and the induction step, we have proven the correctness of the corollary.
\end{proof}

Intuitively, the multi-version forward rules perform a simultaneous transformation of multiple model versions encoded in $SCT_{mv}$. The application of a multi-version rule $L_{mv} \leftarrow K_{mv} \rightarrow R_{mv}$ corresponds to the application of the original rule to all model versions in $P = \bigcap_{v \in V^{L_{mv}}} p(m(v)) \cap \bigcap_{v \in V^{L^{T}_{mv}}} u(m(v))$ and leaves other model versions unchanged. Thus, a multi-version rule application effectively extends the original rule application sequences for versions in $P$ by the associated original rule application, whereas it represents the ``skipping'' of a step for versions not in $P$.

\begin{corollary} \label{cor:mv_max_rule_application_sequence}
For a TGG with associated set of forward rules $\Gamma$ and multi-version forward rules $\Gamma_{mv}$ and a multi-version model triplet $SCT_{mv}$ with bookkeeping and version graph $V$, there is a maximal sequence of rule applications $SCT_{mv} \rightarrow^{\Gamma_{mv}} SCT_{mv}'$ if and only if for all $t \in V^V$, there is a maximal sequence of regular rule applications $SCT_{mv}[t] \rightarrow^{\Gamma} SCT'$ such that $SCT_{mv}'[t] = SCT'$ up to isomorphism including bookkeeping.
\end{corollary}

\begin{proof}
The existence of a sequence of original rule applications for a sequence of multi-version rule applications and all versions $t \in V^V$ and vice-versa is given by Corollary \ref{cor:mv_rule_application_sequence}. From Theorem \ref{the:mv_rule_applicability}, it follows directly that the multi-version sequence is maximal if and only if the regular sequences are maximal for all $t \in V^V$.
\end{proof}

For a deterministic TGG, a correct translation of source graph $S$ is given by any maximal rule application sequence of forward rules that deletes all bookkeeping edges in the source model. Note that because of the determinism criterion, either every maximal rule application sequences or none of them satisfies the bookkeeping criterion.

Thus, for a deterministic TGG and by Theorem \ref{the:correctness_bookkeeping_projection} and Corollary \ref{cor:mv_max_rule_application_sequence}, the results of jointly transforming the model versions using the TGG, that is, the result of repeated application of adapted transformation rules to a multi-version model prepared for multi-version translation until a fixpoint is reached, are equivalent to the results of repeated application of the original rules to the individual model versions prepared for translation.

We thereby have the correctness of the forward transformation using multi-version forward rules $trans^{F}_{mv}$, which applies multi-version forward rules to a multi-version model with bookkeeping until a fixpoint is reached.

\begin{theorem} \label{the:mv_transformation_correctness}
For a correct version history $\Delta^{M_{\{1,...,n\}}}$ and a triple graph grammar with set of forward rules $\Gamma$, it holds up to isomorphism that
$\forall t \in \{1, ..., n\} : trans^{F}_{mv}(init_F(comb(\Delta^{M_{\{1,...,n\}}})), adapt(\Gamma))[t] = trans^F(M_t, \Gamma)$ if $trans^F(M_t, \Gamma)$ contains no bookkeeping edges.
\end{theorem}

\begin{proof}
Follows from Theorem \ref{the:correctness_bookkeeping_projection} and Corollary \ref{cor:mv_max_rule_application_sequence}.
\end{proof}

\section{Incremental Execution of Multi-version Synchronizations (Detailed)} \label{sec:mv_incremental_execution}

As TGGs naturally offer capabilities for incremental synchronization of single-version models \cite{giese2009model}, the related concepts can be transferred to multi-version case to enable direct incremental model synchronization when developing with multi-version models. We therefore consider a standard scenario where new versions of a model can be created, changed and merged by developers and the associated modifications to the model's version history should correctly be propagated to a related target model of a TGG-based transformation. In the following, we discuss how TGGs can be used to react to the different kinds of modifications of multi-version models required in such a scenario.

Formally, the creation of a new version for single-version models corresponds to the introduction of a new model modification $M_i \leftarrow K \rightarrow M_{n+1}$ into a version history $\Delta^{M_{\{1,...,n\}}}$ such that $M_i$ and $M_{n+1}$ are isomorphic, requiring copying of the base version to retain the version history. The realization in the context of multi-version models via a procedure $apply^v_{mv}$ only requires adding a new $version$ node for $M_{n+1}$ along with an incoming $suc$ edge from $M_i$.

The following theorem states that applying this procedure to a previously transformed multi-version model triplet already yields a triplet where all encoded versions of the source model are correctly transformed. Thus, no further synchronization effort is required.

\begin{theorem} \label{the:sync_version_creation}
For a correct version history $\Delta^{M_{\{1,...,n\}}}$, an extended version history $\Delta^{M_{\{1,...,n + 1\}}} = \Delta^{M_{\{1,...,n\}}} \cup \{M_i \leftarrow K \rightarrow M_{n+1}\}$ with $i \in \{1,...,n\}$ and $(V^M_i, E^M_i, s^M_i, t^M_i) = (V^M_{n + 1}, E^M_{n + 1}, s^M_{n + 1}, t^M_{n + 1})$, and a TGG with set of forward rules $\Gamma$, it holds up to isomorphism including bookkeeping that
$\forall t \in \{1, ..., n + 1\} : apply^v_{mv}(SCT_{mv}, M_i \leftarrow K \rightarrow M_{n+1})[t] = SCT_t'$,
with $SCT_{mv} =  trans^{F}_{mv}(init_F(comb(\Delta^{M_{\{1,...,n\}}})), adapt(\Gamma))$ and $M_t \rightarrow^\Gamma SCT_t'$ a maximal rule application sequence.
\end{theorem}

\begin{proof}
Since $apply_v$ only introduces a new $version$ node for version $n + 1$ along with a single incoming $suc$ edge from the $version$ node for version $i$ and does not make any further changes, it follows that $\forall{t' \in \{1, ..., n\}}: apply^v_{mv}(SCT_{mv})[t] = SCT_{mv}[t]$. It also follows that $\forall{v \in V^{apply^v_{mv}(SCT_{mv})}}: ((n + 1)  \in p(v) \leftrightarrow i \in p(v)) \wedge ((n + 1)  \in u(v) \leftrightarrow i \in u(v))$ and hence $apply^v_{mv}(SCT_{mv})[n + 1] = SCT_{mv}[i]$. Thus, because $M_i$ and $M_{n+1}$ are isomorphic, the theorem holds.
\end{proof}

The creation of an element $x$ in a new version $M_i$ in the single-version case formally consists of simply adding the element to the set of nodes or edges of $M_i$ and adjusting source and target functions if $x$ is an edge. In a multi-version model representation, instead a new node $v_{mv}$ of the corresponding adapted type is created and connected to the $version$ node representing $M_i$ via a $cv$ edge. Furthermore, if $x$ is an edge, the related source and target edges are created.

Such a modification to a source model can be synchronized to a target model via the procedures $mark^F_c$ and $trans^F_{mv}$. $mark^F_c$ connects $v_{mv}$ to the $version$ node representing $M_i$ via a $ucv$ edge. $trans^{F}_{mv}$ applies the TGG's multi-version forward rules until a fixpoint is reached. The following theorem states that this yields a multi-version model triplet that correctly encodes the transformation results for all model versions, with $apply^+$ and $apply_{mv}^+$ the procedures for applying a creation modification to the original model respectively the multi-version model.

\begin{theorem} \label{the:sync_element_creation}
For a correct version history $\Delta^{M_{\{1,...,n\}}}$, an element $x$, a version $M_i$ such that $\nexists M_i \leftarrow K \rightarrow M_x \in \Delta^{M_{\{1,...,n\}}}$, and a TGG with set of forward rules $\Gamma$, it holds up to isomorphism including bookkeeping that
$\forall t \in \{1, ..., n\} \setminus \{i\}: trans^{F}_{mv}(mark^F_c(SCT_{mv}', x, i), adapt(\Gamma))[t] = SCT_t'
\wedge \\ trans^{F}_{mv}(mark^F_c(SCT_{mv}', e, i), adapt(\Gamma))[i] = SCT_i'$,
with $SCT_{mv} =  trans^{F}_{mv}(init_F(comb(\Delta^{M_{\{1,...,n\}}})), adapt(\Gamma))$, $SCT_{mv}' = apply_{mv}^{+}(SCT_{mv}, x, i)$, and $M_t \rightarrow^\Gamma SCT_t'$ and $apply^+(M_i, x) \rightarrow^\Gamma SCT_i'$ maximal rule application sequences.
\end{theorem}

\begin{proof}
Since neither $apply_{mv}^{+}$ nor $mark^F_{c}$ and consequently also not $trans^{F}_{mv}$ impact the result of the projection operation for any $t \in \{1, ..., n\} \setminus \{i\}$, the theorem holds for all $t \in \{1, ..., n\} \setminus \{i\}$. For version $i$, it follows from theorem \ref{cor:mv_rule_application_sequence} that there exists a rule application sequence $apply^+(M_i, \delta_+) \rightarrow^\Gamma mark^F_c(SCT'_{mv}, x, i)[i]$, since the projection contains a new, unmarked element corresponding to $x$ and is otherwise unchanged. From corollary \ref{cor:mv_max_rule_application_sequence} hence follows the correctness of the theorem.
\end{proof}

Deletion of an element $x$ from a new version $M_i$ corresponds to the removal of the respective element from the set of nodes or edges of $M_i$ and adjusting source and target functions if $x$ is an edge. To update a related multi-version model accordingly, a $dv$ edge from the node $v_{mv}$ representing $x$ to the $version$ node $m_i$ representing $M_i$ is created.

The procedures $mark^F_{d}$ and $trans^F_{mv}$ synchronize such a modification to a source model to a corresponding target model. $mark^F_{d}$ consists of two steps. First, if $M_i \in u(v_{mv})$, a $udv$ edge between $v_{mv}$ and $m_i$ is created. Otherwise, any correspondence node $c_{mv}$ adjacent to $v_{mv}$ with $M_i \in p(c_{mv})$ as well as any attached target model nodes are connected to $m_i$ via $dv$ edges. Then, this adjustment is transitively propagated to the correspondence node's dependent correspondence nodes that fulfill the presence condition and their attached target model nodes. A correspondence node is dependent on another if the image of the match of the rule application that created the dependent correspondence node contained the required correspondence node. Second, all source model nodes connected to an affected correspondence node other than $v_{mv}$ are connected to $m_i$ via a $ucv$ edge. Finally, multi-version forward rules are applied by $trans^F_{mv}$ until a fixpoint is reached. This yields correct transformation results, as stated by the following theorem, with $apply^-$ and $apply_{mv}^-$ the procedures for applying a deletion modification to the original model respectively the multi-version model.

\begin{theorem} \label{the:sync_element_deletion}
For a correct version history $\Delta^{M_{\{1,...,n\}}}$, an element $x$, a version $M_i$ such that $\nexists M_i \leftarrow K \rightarrow M_x \in \Delta^{M_{\{1,...,n\}}}$, and a TGG with set of forward rules $\Gamma$, it holds up to isomorphism including bookkeeping that
$\forall t \in \{1, ..., n\} \setminus \{i\}: trans^F_{mv}(mark^F_{d}(SCT'_{mv}, x, i), adapt(\Gamma))[t] = SCT_t'
\wedge trans^F_{mv}(mark^F_{d}(SCT'_{mv}, x, i), adapt(\Gamma))[i] = SCT_i'$,
with $SCT_{mv} =  trans^{F}_{mv}(init_F(comb(\Delta^{M_{\{1,...,n\}}})), adapt(\Gamma))$, $SCT_{mv}' = apply_{mv}^{-}(SCT_{mv}, x, i)$, and $M_t \rightarrow^\Gamma SCT_t'$ and $apply^-(M_i, x) \rightarrow^\Gamma SCT_i'$ maximal rule application sequences.
\end{theorem}

\begin{proof}
Since neither $apply_{mv}^{-}$ nor $mark^F_{d}$ impact the result of the projection operation for any $t \in \{1, ..., n\} \setminus \{i\}$, the theorem holds for all $t \in \{1, ..., n\} \setminus \{i\}$. For version $i$, $mark^F_{d}(SCT'_{mv}, e, i)[i]$ no longer contains the element corresponding to $x$, any directly attached or dependent correspondence node, and any related target elements, effectively structurally undoing the rule applications that created these elements. Furthermore, in the projection, all source elements covered by a previously present correspondence node are now marked. Since the projection is otherwise unchanged, it follows from theorem \ref{cor:mv_rule_application_sequence} that there exists a rule application sequence $apply^-(M_i, x) \rightarrow^\Gamma mark^F_{d}(SCT'_{mv}, x, i)[i]$. The correctness of the theorem then follows from corollary \ref{cor:mv_max_rule_application_sequence}.
\end{proof}

Essentially, the synchronization procedures for element creation and deletion thus realize a synchronization similar to the technique described in \cite{giese2009model} in the context of multi-version models.

We consider a merge of two versions for single-version models to be formally represented by the introduction of two new model modifications $M_i \leftarrow K \rightarrow M_{n+1}$ and $M_j \leftarrow K \rightarrow M_{n+1}$ with $M_{n+1} \subseteq M_i \cup M_j$ into a version history $\Delta^{M_{\{1,...,n\}}}$. In a multi-version model encoding, this modification corresponds to the creation of a new $version$ node for $M_{n+1}$, $m_{n+1}$, along with incoming $suc$ edges from the $version$ nodes for $M_i$ and $M_j$. Furthermore, $dv$ edges to $m_{n+1}$ from all nodes representing elements in $M_i$ and $M_j$ that are not in $M_{n+1}$ are created.

Synchronization of such a change to a source version history is achieved via the procedures $mark^F_{m}$ and $trans^F_{mv}$. $mark^F_{m}$ performs three steps: First, similarly to how deletion modifications for new versions are handled, any correspondence node $c_{mv}$ with $M_{n+1} \in p(c_{mv})$ adjacent to a node $v_{mv}$ representing a deleted element, as well as any attached target model nodes are connected to $m_{n+1}$ via $dv$ edges and this adjustment is transitively propagated to the correspondence node's dependent correspondence nodes that fulfill the presence condition and their attached target model nodes. For any such $v_{mv}$ with $M_{n+1} \in u(v_{mv})$, a $udv$ edge between $v_{mv}$ and $m_{n+1}$ is added. Second, for each source model element that would be connected to more than one correspondence node in the projection to $n + 1$ (one of which has to be present in the projection to $i$ and the other in the projection to $j$), a $dv$ edge to $m_{n+1}$ is added to the correspondence node present in the projection to $j$ and this adjustment is again propagated to dependent correspondence nodes and associated target model elements. Third, all source model nodes $v'_{mv}$ connected to any affected correspondence node with $m_{n+1} \in p(v'_{mv})$ are connected to $m_{n+1}$ via a $ucv$ edge. $trans^F_{mv}$ then applies the TGG's multi-version forward rules until a fixpoint is reached. This yields a correct multi-version encoding of the transformation results for all model versions, with $apply^m_{mv}$ the procedure for applying the merge modification to the multi-version model.

\begin{theorem} \label{the:sync_version_merge}
For a correct version history $\Delta^{M_{\{1,...,n\}}}$, an extended version history $\Delta^{M_{\{1,...,n + 1\}}} = \Delta^{M_{\{1,...,n\}}} \cup \{M_i \leftarrow K_i \rightarrow M_{n+1}, M_j \leftarrow K_j \rightarrow M_{n+1}\}$ with $i, j \in \{1,...,n\}$, $i \neq j$, and $M_{n+1} \subseteq M_i \cup M_j$, and a TGG with set of forward rules $\Gamma$, it holds up to isomorphism including bookkeeping that
$\forall t \in \{1, ..., n + 1\} : trans^F_{mv}(mark^F_m(SCT'_{mv}, M_i \leftarrow K_i \rightarrow M_{n+1}, M_j \leftarrow K_j \rightarrow M_{n+1}), adapt(\Gamma))[t] = trans^F(M_t, \Gamma)$,
with $SCT_{mv} =  trans^{F}_{mv}(init_F(comb(\Delta^{M_{\{1,...,n\}}})), adapt(\Gamma))$, $SCT'_{mv} = apply^m_{mv}(SCT_{mv}, M_i \leftarrow K_i \rightarrow M_{n+1}, M_j \leftarrow K_j \rightarrow M_{n+1})$, and $M_t \rightarrow^\Gamma SCT_t'$ a maximal rule application sequence.
\end{theorem}

\begin{proof} 
Since $apply^m_{mv}$ only introduces a new $version$ node for version $n + 1$ along with two incoming $suc$ edges from the $version$ nodes for versions $i$ and $j$ and hence neither $apply^m_{mv}$ nor $mark^F_{m}$ and consequently also not $trans^F_{mv}$ impact the result of the projection operation for any $t' \in \{1, ..., n\}$, it follows that $\forall t \in \{1, ..., n\} : trans^F_{mv}(mark^F_m(SCT'_{mv}, M_i \leftarrow K_i \rightarrow M_{n+1}, M_j \leftarrow K_j \rightarrow M_{n+1}), adapt(\Gamma))[t] = trans^F(M_t, \Gamma)$.

Furthermore, since $dv$ edges to the $version$ node corresponding to $M_{n+1}$ are transitively added to any attached correspondence node that might otherwise be present in $mark^F_m(SCT'_{mv}, M_i \leftarrow K_i \rightarrow M_{n+1}, M_j \leftarrow K_j \rightarrow M_{n+1})[n+1]$, dependent correspondence nodes, and associated target elements, if any of its associated source elements is not present anymore, it essentially structurally undoes the related forward rule applications that created these elements. It thus follows that for any correspondence node and target element remaining in $mark^F_m(SCT'_{mv}, M_i \leftarrow K_i \rightarrow M_{n+1}, M_j \leftarrow K_j \rightarrow M_{n+1})[n + 1]$, there must exist a sequence of applications of rules from $\Gamma$ that creates the correspondence node and its target elements.

Furthermore, since $dv$ edges are similarly added for redundant correspondence nodes, it follows that no two correspondence nodes and associated target elements in $mark^F_m(SCT'_{mv}, M_i \leftarrow K_i \rightarrow M_{n+1}, M_j \leftarrow K_j \rightarrow M_{n+1})[n + 1]$ can be created by separate rule applications that require the deletion of new bookkeping edges for the same source element. Since rules from $\Gamma$ are productions except for the bookkeeping mechanism, the application of one rule can only prevent the application of another rule via the bookkeeping mechanism. Thus, because of the adjustment of the set $u(v)$ for any impacted node $v \in V^{SCT_{mv}}$, there must exist a sequence of rule applications $M_{n+1} \rightarrow^R apply^m_{mv}(SCT_{mv}, M_i \leftarrow K_i \rightarrow M_{n+1}, M_j \leftarrow K_j \rightarrow M_{n+1})[t]$.

From corollary \ref{cor:mv_max_rule_application_sequence} then follows the correctness of the theorem.
\end{proof}

From the correctness of the synchronization procedures for the considered types of modifications to a version history follows that a sequence of such modifications can be handled via a procedure $sync^F_{mv}$, which performs the appropriate synchronization for each modification in the sequence in order, with $apply$ the procedure for applying a sequence of modifications to the original model one after another.

\begin{theorem}
For a correct version history $\Delta^{M_{\{1,...,n\}}}$, a sequence of version creation, element creation, element deletion, and merge modifications $S_\delta$ creating new versions $M_{n+1}, ..., M_{n + m}$, and a TGG with set of forward rules $R$, it holds up to isomorphism that
$\forall t \in \{1, ..., n + m\} : sync^F_{mv}(SCT_{mv}, S_\delta, adapt(\Gamma))[t] = trans^F(M'_t, \Gamma)$ if $trans^F(M_t, \Gamma)$ contains no bookkeeping edges,
with $SCT_{mv} =  trans^{F}_{mv}(init_F(comb(\Delta^{M_{\{1,...,n\}}})), adapt(\Gamma))$ and $\Delta^{M_{\{1, ..., n + m\}}} = apply(\Delta^{M_{\{1,...,n\}}}, S_\delta)$.
\end{theorem}

\begin{proof}
Follows from the correctness of theorems \ref{the:sync_version_creation}-\ref{the:sync_version_merge} and the determinism property for TGGs.
\end{proof}

\section{Evaluation} \label{sec:evaluation}

In order to evaluate our approach empirically with respect to execution time performance and memory consumption, we have realized the presented concepts in the MoTE2 tool \cite{hildebrandt2014} for TGG-based model transformation, which is implemented in the context of the Java-based Eclipse Modeling Framework \cite{emf} and has been shown to be efficient compared to other model transformation tools \cite{hildebrandt2014}.

As an application scenario, we consider the transformation of Java abstract syntax graphs to class diagrams. We have therefore modeled this transformation as a TGG with MoTE2 and use the original and our adapted implementation to automatically derive forward rules respectively multi-version forward rules according to Section \ref{sec:mv_rules_derivation}, as well as the single-version forward synchronization rules employed by MoTE2.

To obtain realistic source models, we have extracted the version history of one small Java project (\emph{rete}, about 60 versions) and one larger open source Java project (\emph{henshin} \cite{arendt2010henshin}, about 2600 versions) from their respective Git \cite{git} repositories and have constructed the corresponding history of the related abstract syntax graphs using the MoDisco tool \cite{bruneliere2010modisco}. As input for the solution presented in Sections \ref{sec:mv_rules_derivation}, \ref{sec:mv_transformation_execution}, and \ref{sec:mv_incremental_execution}, we have consolidated both version histories into multi-version models using a mapping based on hierarchy and naming.\footnote{Our implementation and datasets are available under\cite{implementation} respectively \cite{barkowsky_matthias_2023_8109856}.} Based on this, we experiment with two application scenarios for model transformation in the context of models with multiple versions.

\subsection{Batch Transformation Scenario}

First, we consider a batch transformation scenario, where all versions of the Java abstract syntax graph are to be translated to their corresponding class diagram. This scenario emulates a situation where a new model along with a transformation to derive it from an existing model is introduced to an ongoing development process.

We therefore run the following model transformations for both repositories and measure the overall execution time and memory consumption of the involved models\footnote{All experiments were performed on a Linux SMP Debian 4.19.67-2 machine with Intel Xeon E5-2630 CPU (2.3\,GHz clock rate) and 386\,GB system memory running OpenJDK version 11.0.6. Reported execution time and memory consumption measurements correspond to the mean result of 10 runs of the respective experiment. Memory consumption measurements were obtained using the Java Runtime class.}:

\begin{itemize}
\item \textbf{SVM B}: individual forward transformation of all model versions in the version history using the original MoTE2 implementation
\item \textbf{MVM B}: joint forward transformation of all model versions in the version history using a multi-version model encoding and our implementation of the technique presented in Sections \ref{sec:mv_rules_derivation} and \ref{sec:mv_transformation_execution}
\end{itemize}

Note that the SVM B strategy would require initial projection operations and a final combination of transformation results to work within the framework of multi-version models. However, for fairness of comparison of the transformation, we do not consider these additional operations in our evaluation. We however do consider initialization effort for each model version required by the MoTE2 engine.

To investigate scalability, we also execute the transformations for subsets of the version history, that is, only for a subset of all model versions in the case of SVM B and for a multi-version model encoding only a subset of all versions in the case of MVM B. Figure \ref{fig:times_batch} shows the execution times of the transformations using the two strategies for subsets of different size and how the measured execution times are composed of time required to initialize the MoTE2 engine (init), indexing the models for pattern matching by the employed pattern matching tool (index), and actual execution of forward rules (execute). Except for the two smaller subsets of the version history of the smaller repository, the transformation based on multi-version models requires less time than the transformation of the individual model versions using the original MoTE2 tool, with the most pronounced improvement for the full history of the larger repository.

\begin{figure}
\centering
\includegraphics[width=0.8\linewidth]{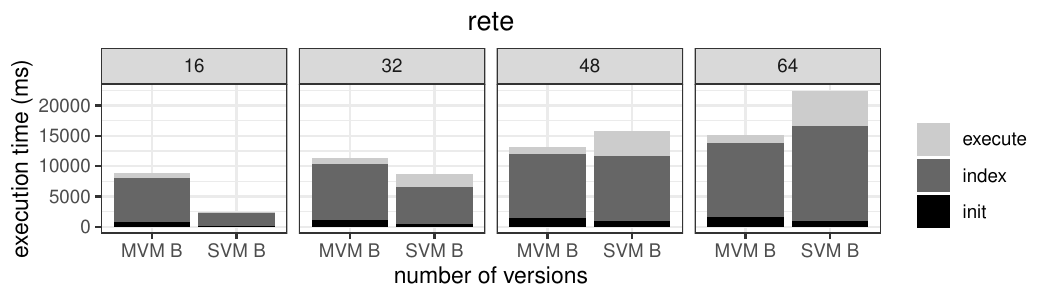}

\includegraphics[width=0.8\linewidth]{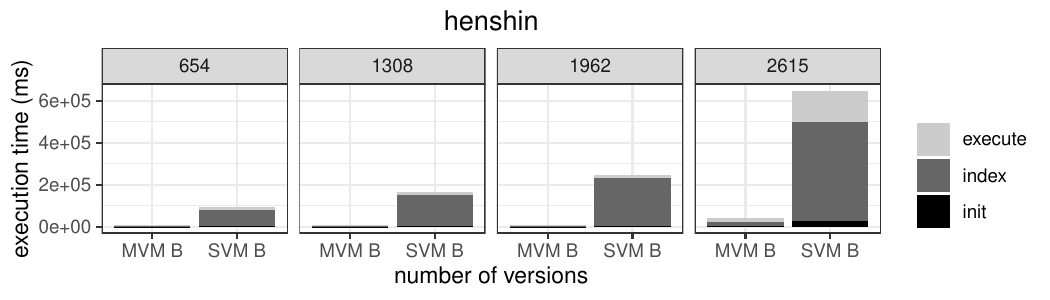}
\caption{Execution time measurements for the batch scenario}\label{fig:times_batch}
\end{figure}

The improvement in efficiency and scalability for larger version histories can be explained by the fact that many elements in the abstract syntax graphs of the repositories are shared between many versions. SVM B has to perform a separate transformation, including separate pattern matching and indexing, for each model version. In contrast, MVM B only performs a transformation including pattern matching over a single multi-version model, along with efficient search operations over the version graph, and only requires indexing once. For the larger multi-version models, this effect outweighs the higher initialization effort for the MVM variant and the increase in pattern matching and indexing effort resulting from the technically less efficient encoding of edges and attributes as nodes in the multi-version model.

To investigate the memory consumption of the different encodings of the model triplets produced by the transformation, we stored the produced triplets and loaded them into memory in a second experiment. The resulting memory measurements can be found in Figure \ref{fig:memory}, which shows that the multi-version encoding was more compact than the corresponding na\"ive encoding for all considered subsets of the larger repository's history, while it required more memory for the three smaller subsets of the smaller repository's version history. The overhead can be explained by the less efficient encoding of edges and attributes as nodes in a multi-version model, which for larger histories is however outweighed by the reduction in redundancy. With a memory consumption of about 200\,MB and 400\,MB, the multi-version model encoding in both cases is compact enough to fit into a regular PC's main memory.

\begin{figure}
\centering
\includegraphics[width=0.8\linewidth]{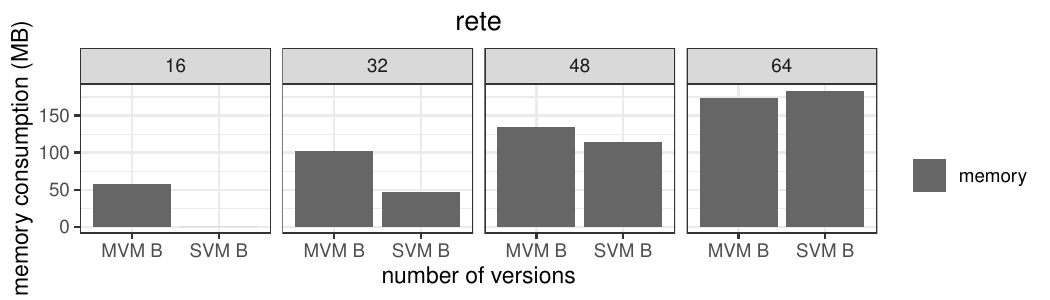}

\includegraphics[width=0.8\linewidth]{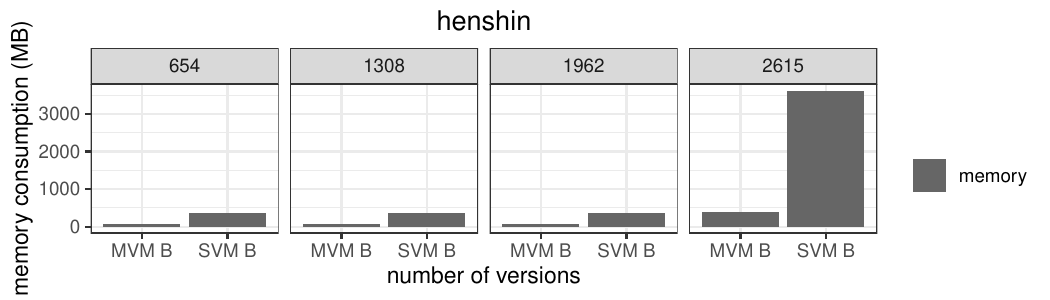}
\caption{Memory consumption measurements}\label{fig:memory}
\end{figure}

Notably, memory consumption for the three smaller subsets of the larger repository's version history is very similar. This is due to the fact that many of the models for versions between version 600 and 2000 are actually almost empty as the repository does not contain a project with the given name. This is also the case for the first 15 versions in the smaller repository, which causes memory consumption for the smallest version history subset to fall below the accuracy threshold of the methods used to measure it. However, due to a quirk of the MoTE2 engine, which performs pattern matching over the input TGG and the engine itself, even these empty models cause some indexing and initialization effort, which explains the execution time measurements in Figure \ref{fig:times_batch}.

The measurements overall indicate that at least in our current implementation, encoding edges and attributes as nodes in a multi-version model causes substantial overhead regarding both execution time and memory consumption. However, the results also confirm that joint transformation of all encoded versions can significantly improve performance compared to the separate transformation of each individual model version for large histories with many shared elements between model versions, more than compensating for the overhead caused.

\subsection{Incremental Synchronization Scenario}

As a second scenario, we consider the incremental synchronization of changes between pairs of versions of the abstract syntax graphs. This scenario aims to emulate an ongoing development process where new versions of the abstract syntax graph are iteratively produced by user edits and the corresponding version history of the corresponding class diagram should continuously be updated to reflect the newly introduced versions.

We therefore incrementally rebuild the final models in the repositories' histories, starting with the initial version and iteratively applying the changes of the successor versions according to a topological sorting of the version DAG. In case of branching development, we split execution and consider each branch separately until the branches are merged again. After an initial batch transformation of the root version, we perform an incremental synchronization after the integration of each successor version and measure the related execution time. For synchronization, we consider the following techniques:

\begin{itemize}
\item \textbf{SVM I}: forward synchronization of changes modifying a direct predecessor version into each version in the history via the regular single-version synchronization by MoTE2
\item \textbf{MVM I}: forward synchronization of changes resulting from the iterative integration of each version in the history into a multi-version model encoding via the multi-version synchronization approach from Section \ref{sec:mv_incremental_execution}
\end{itemize}

Note that we assume that the version history of the involved models is to be preserved, that is, the introduction of a new version must not override the base version. For the synchronization working with single-version models, this means that in order to create a new version of source and target model, the base version of both models along with the connecting correspondence model has to be copied before changes can be applied and the synchronization can be executed. We therefore consider the time required for this copying effort in the execution time measurements for SVM I. In contrast, this is not required for the MVM I strategy, which natively preserves the full version history by only allowing the types of changes described in Section \ref{sec:mv_incremental_execution}.

The aggregate execution times for synchronizing all versions up to the n-th version in the topological sorting of the version graph are plotted in Figure \ref{fig:times_incremental}. For the smaller repository, SVM I outperforms MVM I by about factor 1.5 even after considering the required copying. For the larger repository, the execution time of MVM I is comparable to the execution time of SVM I including copying.

The plots demonstrate that SVM I requires some computational effort even for versions without changes due to the copying of the base version. In principle, MVM I only has to perform computations for actual changes of the source model, but requires substantially more time to process such a change than SVM I due to the less efficient encoding of nodes and edges and the fact that in some cases, synchronization of a change requires traversal of larger parts of the version DAG, potentially causing effort in the version DAG's size. To reduce the required effort for the latter, we employ a simple indexing structure for the version DAG, the updating of which theoretically causes effort in the version DAG's size whenever a new version is introduced. However, the required execution time was below the granularity threshold of milliseconds in our experiments and is hence not visible in the plots.

Lastly, even without a significant gain in performance for incremental synchronization, the usage of the multi-version model encoding of the version histories in both the batch and incremental scenario means that analyses that are specific to this representation or more efficient there, such as those presented in \cite{Barkowsky2022}, can directly be executed over the transformation results without requiring a change of encoding.

\begin{figure}
\centering
\includegraphics[width=0.8\linewidth]{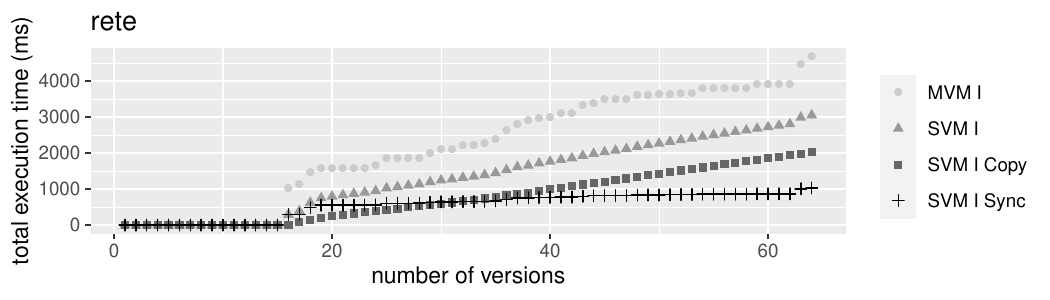}

\includegraphics[width=0.8\linewidth]{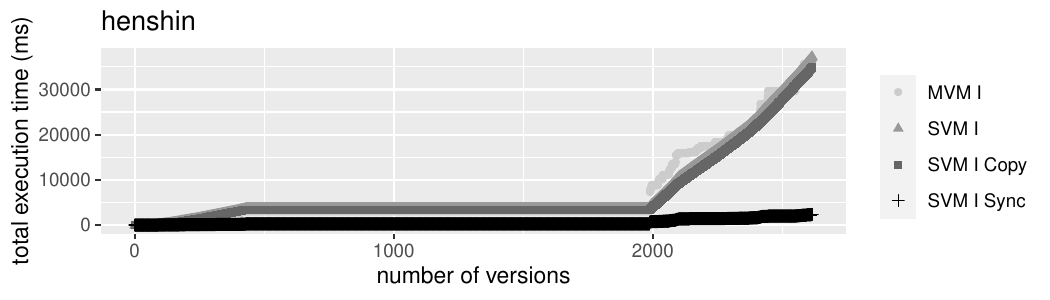}
\caption{Execution time measurements for the incremental scenario (larger repository)}\label{fig:times_incremental}
\end{figure}

\subsection{Threats to Validity}

Threats to the internal validity of our experimental results include unexpected behavior of the Java virtual machine such as garbage collection. To address this threat, we have performed multiple runs of all experiments and report the mean measurement result, with the standard deviation of overall execution time and memory consumption always below 10\% of the mean value. To minimize the impact of the concrete implementation, we have realized our solution in the framework of the transformation tool we use for comparison and thereby largely use the same execution mechanism.

Measuring memory consumption of Java programs is known to be a challenge. While we attempted to improve the reliability of these measurements by performing multiple runs of each experiment and suggesting to the JVM to perform garbage collection before every measurement, the reported results are not necessarily accurate but can only serve as an indicator.

To mitigate threats to external validity, we use real-world models as the source models of the transformation. However, we remark that our results are not necessarily generalizable to different examples or application domains and make no quantitative claims regarding the performance of our approach.

\section{Related Work} \label{sec:related_work}

The general problem of model versioning has already been studied extensively, both formally \cite{diskin2009model,rutle2009category} and in the form of concrete tool implementations \cite{murta2008towards,koegel2010emfstore}. Several solutions employ a unified representation of a model's version history similar to multi-version models \cite{rutle2009category,murta2008towards}. However, due to the problem definition focusing on the management of different versions of a single model, model transformation based on a unified encoding is out of scope for these approaches.

There is also a significant body of previous work on synchronization of concurrently modified pairs of models using triple graph grammars \cite{xiong2013synchronizing,orejas2020incremental}. The focus of these works is the derivation of compatible versions of source and target model that respect the modifications to either of them. This paper aims to make a step in an orthogonal direction, namely towards allowing living with inconsistencies by enabling developers to temporarily work with multiple modified, possibly conflicting versions of source and target model.

Furthermore, there exist several approaches for optimizing the performance of incremental model synchronization with TGGs, for instance \cite{hildebrandt2014} and \cite{fritsche2021avoiding}. Since the underlying concepts are mostly orthogonal to the ideas related to multi-version models, an integration into the approach proposed in this paper is an interesting direction for future work.

In the context of software product lines, so-called 150\% models are employed to encode different configurations of a software system \cite{10.1007/11561347_28,10.1145/3377024.3377030}. In this context, Greiner and Westfechtel present an approach for propagating so-called variability annotations along trace links created by model transformations \cite{westfechtel2020extending}, explicitly considering the case of transformations implemented via TGGs. However, not integrating this propagation with the transformation process would mean that certain cases that are covered by our approach could not be handled, for instance if a model element would be translated differently in different model versions based on its context.

The joint execution of queries over multiple versions of an evolving model has been considered for both the case with \cite{Barkowsky2022} and without \cite{giese2019metric,sakizloglou2021incremental,garcia2019querying} parallel, branching development. This paper builds on these results, but instead of focusing on pure queries without side-effects considers the case of writing operations in the form of model transformations.

\section{Conclusion}\label{sec:conclusion}

In this paper, we have presented a step in the direction of model transformation and synchronization for multi-version models in the form of an adaptation of the well-known triple graph grammar formalism that enables the joint transformation of all versions encoded in a multi-version model as well as synchronization of subsequent updates. The presented approach is correct with respect to the translation semantics of deterministic triple graph grammars for individual model versions, that is, it produces equivalent results. Initial experiments for evaluating the efficiency of our approach demonstrate that our technique can improve performance of the transformation compared to a na\"ive realization, but can also cause significant computational overhead especially in the synchronization case, in a realistic application scenario.

In future work, we want to explore the possibility of improving the efficiency of multi-version model transformations via incremental pattern matching for multi-version models. Another interesting direction is the integration of advanced application conditions for the specification of triple graph grammar rules, such as nested graph conditions, into our approach. Finally, a more extensive evaluation can be conducted to further study the performance of the presented technique.

\section*{Acknowledgements}

This work was developed mainly in the course of the project modular and incremental Global Model Management (project number 336677879) funded by the DFG.

\clearpage

\bibliographystyle{IEEEtran}
\bibliography{references}

\end{document}